\documentclass[a4paper,11pt]{article}

\usepackage{amsthm}
\usepackage{graphicx, amsmath, amssymb}

\usepackage{amssymb}
\theoremstyle{plain}
  \newtheorem{theorem}{Theorem}
  \newtheorem{lemma}[theorem]{Lemma}
  \newtheorem{corollary}[theorem]{Corollary}

\usepackage{vmargin}
\setmarginsrb{2.54cm}{2.54cm}{2.54cm}{2.54cm}{0pt}{0pt}{0pt}{6mm}

\newcommand{\tw}{\textbf{tw}}

  \newcommand{\cO}{\mathcal{O}}

  \newcommand{\constPMC}{1.734601}
  
  \newcommand{\constOldPMC}{1.7549}

  \newcommand{\cC}{\mathcal{C}}

  \newcommand{\pmc}{potential maximal clique}
  \newcommand{\pmcs}{potential maximal cliques}

\newcommand{\sm}{\setminus}

\title{Finding Induced Subgraphs via Minimal Triangulations}

\author{Fedor V. Fomin\thanks{Department of Informatics, University of Bergen, Norway.  
\texttt{\{fedor.fomin|yngve.villanger\}@ii.uib.no.} Partially supported by the Norwegian Research Council. }  \and \addtocounter{footnote}{-1} Yngve Villanger\footnotemark}

\date{}
\begin{document}

\maketitle
\begin{abstract}
\noindent Potential maximal cliques and minimal separators are combinatorial objects 
which were introduced and studied in the realm of  minimal triangulation problems 
including Minimum Fill-in and Treewidth. We discover unexpected applications of these 
notions to the field of moderate exponential algorithms.  In particular, we show that 
given an n-vertex graph G together with its set of potential maximal cliques, and an integer t, 
it is possible in time the number of potential maximal cliques times $O(n^{O(t)})$ 
to find a maximum induced subgraph of treewidth t in G and 
for a given graph F of treewidth t, to decide if G contains an induced subgraph isomorphic to F. 
Combined with an improved algorithm enumerating all potential maximal cliques in time $O(1.734601^n)$, 
this yields that both the problems are solvable in time  $1.734601^n$ * $n^{O(t)}$. 
% 
% 
% 
% The abstract must precede the maketitle command.  Be sure
%   not to issue the maketitle command twice!  The abstact should
%   consist of plain text, without mathematical formul\ae, and it must
%   be a self-standing text (it will appear separately in collections of
%   abstracts): in the abstract, do not use coded bibliographical
%   references such as [33], use rather (Emerson, Lake and Palmer,
%   1973).
\end{abstract}

\maketitle

\section{Introduction}

One of the most fundamental problems in Graph Algorithms  is, for a given graph $G=(V,E)$, to find a maximum or minimum subset $S$ of $V$ that satisfies some property $\Pi$. For example, when $S$ is required to be a maximum set of pairwise adjacent vertices  this is the \textsc{Maximum Clique} problem. When $S$ is required to be a maximum set of pairwise non-adjacent vertices this is the 
\textsc{Maximum Independent Set}  problem.  Its complement, the \textsc{Minimum Vertex Cover} problem, is to  find a minimum set $S$ such that the graph $G\setminus S$ is an independent set.  
Another examples are \textsc{Maximum Induced Forest}, where one is seeking for a set of vertices inducing a forest of maximum size, or its complement  \textsc{Minimal  Feedback Vertex Set} which is to remove the minimum number of vertices to destroy all cycles.

All these examples are special cases of the problem, where one seeks a maximum subset of vertices that induces a subgraph of $G$ from some given graph class $\mathcal{C}$.
  If $G$ is an $n$-vertex graph, and recognition of  graphs from $\mathcal{C}$ can be done in polynomial time, then the trivial brute force algorithm   solves the problem in time $2^n n^{\cO(1)}$. One of the crucial  questions in the area of moderate exponential algorithms is if the brute force algorithm can be avoided to solve any hard (NP-hard, $\# P$, PSPACE-hard, etc.) problem. So far we are still very far from answering this question. For some problems we know  how to avoid the brute force search, and for some problems, like SAT, it is a big open problem in the area. Similar situation is with the problem of finding a maximum induced subgraph from a given class  $\mathcal{C}$.
For some simple graph classes  $\mathcal{C}$  the trivial $2^n$-barrier has been broken. The most well studied case is when  $\mathcal{C}$ 
is the class of graphs without edges, or the class of graphs of treewidth $0$. In this case, we are looking for an independent set of maximum size.  This is the classical NP-hard problem and it is  well studied in the realm of  moderate exponential  algorithms. The classical result of Moon and Moser \cite{MoonM65} (see also  Miller and Muller \cite{MillerMuller60}) from the 1960s   can be easily turned into algorithms finding a maximum independent set in time $3^{n/3}n^{\cO(1)}$. 
Tarjan and Trojanowski \cite{TarjanTrojanowski77}  gave a $O(2^{n/3})$ time algorithm. There were several non-trivial steps in improving the running time of the algorithm including 
the work of Jian~\cite{Jian86}, Robson~{\cite{Robson86}, and Grandoni et al. \cite{FominGK06soda}.
A significant amount of  research was also devoted to algorithms for the 
\textsc{Maximum Independent Set} problem on sparse
graphs, some examples are~\cite{
%Beigel99,
%ChenKJ01,
ChenKX05,
Furer06,
Razgon06-is}.
It is easy to show  that a simple branching algorithm can compute a maximum induced path or cycle in time $3^{n/3}n^{\cO(1)}$. However, 
breaking the  $2^n$-barrier even  for the case when the class $\mathcal{C}$  is a forest, i.e. the class of graphs of treewidth $1$, was an open problem in the area until very recently.   
The first exact algorithm breaking the trivial  $2^n$-barrier is due
to Razgon  \cite{Razgon06}. 
 The running time $\cO(1.8899^n)$ of the algorithm from \cite{Razgon06} was improved in \cite{FominGaspersPyatkin06,FominGPR08-On} to
$\cO(1.7548^n)$. All these  algorithms for  \textsc{Maximum Independent Set}  and  \textsc{Maximum Induced Forest}  are 
so-called branching algorithms (a variation of
Davis-Putnam-style exponential-time backtracking  \cite{DavisP60-A}). 
There is also a relevant work of  Gupta et al. \cite{GuptaRS06}
 who used branching to show that for every fixed $r$, there are at most $c^n$ $r$-regular subgraphs for some $c<2$. For example, for  
 \textsc{Maximum Induced Matching} and \textsc{Maximum 2-Regular Induced Subgraph}, their results yield algorithms solving these problems in time 
 $\cO(1.695733^n)$ and $\cO(1.7069^n)$, respectively. However, the results of Gupta et al. strongly depend on the regularity of the maximum subgraphs. 
To our knowledge,  prior to our work no algorithms better than the trivial brute-force $\cO(2^n)$ were  known  for more complicated  classes $\mathcal{C}$.

In this work we make a step aside the ``branching"  path and use a completely different approach for problems on finding induced subgraphs. 
Our approach is based on a tools from the area of minimal triangulations, namely, potential maximal cliques. 
Minimal triangulations  are the result of adding an inclusion minimal set of edges to produce 
a triangulation (or chordal graph). The study of  minimal triangulations dates back to the 1970s and originated from 
research on  
sparse matrices and vertex elimination in graphs. Minimal separators are one of the main tools in the study of minimal triangulations. We refer to the survey of Heggernes 
\cite{Heggernes06} for more information on triangulations. 
Potential maximal cliques were defined by 
Bouchitt\'e and    Todinca \cite{BouchitteT01,BouchitteT02} and were used in different algorithms for computing the treewidth of a graph
\cite{FominKratschTodincaV08,FominV08}. A subset of vertices $C$  of a graph $G$ is a \emph{potential maximal clique} if there is a minimal triangulation $TG$ of $G$ such that  $C$ is a maximal clique in $TG$.  At first glance it is not clear, what is the relation between potential maximal cliques and induced subgraphs. Our first main result establishes such a  relation.

\begin{itemize}
\item Let $\Pi_G$ be the set of potential maximal cliques in $G$.  A maximum induced subgraph of treewidth $t$ in an $n$-vertex graph $G$ can be found in time $\cO(|\Pi_G| \cdot n^{\cO(t)})$ (Section~\ref{sec:_hidden_decomposition}).
\end{itemize}
As we already mentioned, the well studied \textsc{Maximum Independent Set} (and its dual \textsc{Minimum Vertex Cover}) and \textsc{Maximum Induced Forest}  (and \textsc{Minimum Feedback Vertex Set}) are the special cases for $t=0$ and 
$t=1$, respectively.   
Our second  main  result shows that
\begin{itemize}
\item  All potential maximal cliques can be enumerated in time  $\cO(\constPMC^n)$ (Section~\ref{sec:faster_pmc}).
\end{itemize}

Combining both results, we obtain that  a maximum induced subgraph of treewidth $t$ in an $n$-vertex graph $G$ can be found in time $\cO( \constPMC^n \cdot n^{\cO(t)})$. 
While for $t=0$ (the case  of \textsc{Maximum Independent Set}) the existing branching algorithms are much faster than $\cO( \constPMC^n)$, already for   $t=1$ (the case of \textsc{Maximum Induced Forest}) our algorithm is already  faster than the best known branching algorithm \cite{FominGPR08-On}. For fixed $t\geq 2$, no algorithm better than  the trivial $\cO(2^n n^{\cO(1)})$ brute force algorithm was known.

With small modifications, our algorithm can be used for other problems involving induced subgraphs. 
As an example, we show how to solve the induced subgraph isomorphism problem, which is to decide if $G$ contains an induced subgraph isomorphic to a given graph $F$ (Section~\ref{sec:modification}). We show that when the treewidth of $F$ is at most $t$, then this problem is solvable in time 
 $ \constPMC^n \cdot n^{\cO(t)}$.  In particular, when the treewidth of $F$ is $o(n/\log n)$, for example when $F$ is a planar graph, or a graph excluding some fixed graph as a minor, the running time of our algorithm is  $ \constPMC^{n+o(n)}$.
 Let us note that no algorithm faster than the trivial brute-force algorithm was known even when $F$ is a tree. 

Finally, our new algorithm enumerating potential maximal cliques is not only  (slightly) faster than the algorithm from \cite{FominV08}  and thus by \cite{FominKratschTodincaV08}, directly implies faster exact algorithm computing the treewidth of a graph. It is also significantly simpler than the previous algorithms and is easy to implement. 
Due to space limitations, some proofs are omitted. 
A full version will appear at some later point. 
%For a complete version with all the proofs see \cite{FominV09}.

\section{Preliminaries}
We denote by $G=(V,E)$ a finite, undirected, and simple graph with
$|V|=n$ vertices and $|E|=m$ edges. For any nonempty subset
$W\subseteq V$, the subgraph of $G$ induced by $W$ is denoted by
$G[W]$. For $S \subseteq V$ we often use $G \sm S$ to denote $G[V
\sm S]$. The \emph{neighborhood} of a vertex $v$ is $N(v)=\{u\in
V:~\{u,v\}\in E\}$, $N[v] = N(v) \cup \{v\}$, and for a vertex set
$S \subseteq V$ we set $N(S) = \bigcup_{v \in S} N(v)\sm S$, $N[S]
= N(S) \cup S$. A \emph{clique} $C$ of a graph $G$ is a subset of
$V$ such that all the vertices of $C$ are pairwise adjacent. By
$\omega (G)$ we denote the maximum clique-size of a graph~$G$.

A graph $H$ is {\em chordal} (or {\em triangulated}) if every
cycle of length at least four has a chord, i.e., an edge between
two nonconsecutive vertices of the cycle. A {\em triangulation} of
a graph $G=(V,E)$ is a chordal graph $H = (V, E')$ such that $E
\subseteq E'$. Graph $H$ is a {\em minimal triangulation} of $G$ if
 for every
edge set $E''$ with $E \subseteq E'' \subset E'$, the
graph $F=(V, E'')$ is not chordal.

The notion of
treewidth is due to Robertson and Seymour \cite{RobertsonS86}. A {\em
tree decomposition} of a graph $G=(V,E)$, denoted by $TD(G)$, is a
pair $({X}, T)$ in which $T=(V_T, E_T)$ is a tree and ${X}=\{{X}_i
\mid i\in V_T\}$ is a family of subsets of $V$, called  \emph{bags},   such that
\begin{itemize}
\item[(i)] $\bigcup_{i\in V_T}{X}_i= V$;
\item[(ii)] for each edge $e=\{u,v\} \in E$ there  exists an $i\in V_T$ such that both $u$ and $v$
belong to ${X}_i$; 
\item[(iii)] for all $v\in V$, the set of nodes
$\{i\in V_T \mid v \in {X}_i\}$ induces a connected subtree of $T$.
\end{itemize}
The maximum  of $|{X}_i|-1$, $i\in V_T$,  is called the {\em width} of the
tree decomposition. The {\em treewidth} of a graph $G$, denoted by $\tw(G)$,
is the minimum width taken over all tree decompositions of $G$.

\begin{theorem}[folklore]\label{thm:folk}
For any graph  $G$,  $\tw(G)\le k$ if and only if there is a triangulation
$H$ of $G$ such that $\omega(H)\le k+1$.
\end{theorem}

Let $u$ and $v$ be two non adjacent vertices of a graph $G=(V,E)$. A set of
vertices $S \subseteq V$
is a {\em $u,v$-separator} if $u$ and $v$  are in different connected components
of the graph $G[V \sm S]$. A connected component $C$ of $G[V \sm S]$
is a {\em full} component associated to $S$ if $N(C)=S$. Separator 
$S$ is a {\em minimal
$u,v$-separator} of $G$ if no proper subset of $S$ is a $u,v$-separator.
Notice that a minimal separator can be
strictly included in another one. We denote by $\Delta_G$ the set of all
minimal separators of $G$.

A set of vertices $\Omega \subseteq V$ of a graph $G$ is called a
{\em potential maximal clique} if there is a minimal triangulation
$H$ of $G$ such that $\Omega$ is a maximal clique of $H$. We
denote by $\Pi_G$ the set of all potential maximal cliques of $G$.

For a minimal separator $S$ and a full connected component  $C$ %$C \in \mathcal C(S)$
 of $G \setminus S$, 
 we say that $(S,C)$ is a {\em block} associated to
$S$. We sometimes use the notation $(S,C)$ to denote the set of
vertices $S \cup C$ of the block. It is easy to see that if $X
\subseteq V$ corresponds to the set of vertices of a block, then
this block $(S,C)$ is unique: indeed, $S = N(V \sm X)$ and $C=X
\sm S$.

We also need the following result of Bouchitt{\'e} and Todinca on the structure of potential maximal cliques.

\begin{theorem}[Bouchitt{\'e} and Todinca \cite{BouchitteT01}]\label{th:pmc_sep}
Let $K \subseteq V$ be a set of vertices of the graph $G=(V,E)$.
Let $\cC(K) = \{ C_1, \ldots, C_p\}$ be the set of 
connected components of $G \sm K$ and let ${\mathcal S}(K) = \{
S_1,S_2, \ldots  , S_p\}$, where $S_i = N(C_i)$, $i\in
\{1,2,\ldots ,p\}$, is the set of those vertices of $K$ which are
adjacent to at least one vertex of the component $C_i$. Then
$K$ is a potential maximal clique of $G$ if and only if
\begin{enumerate}
\item[{\rm 1.}] $G \sm K$ has no full component associated to $K$, and
\item[{\rm 2.}] the graph on the vertex set $K$ obtained from $G[K]$ by
completing each $S_i \in {\mathcal{S}}(K)$ into a clique is a
complete graph.
\end{enumerate}
Moreover, if $K$ is a potential maximal clique, then
$\mathcal S(K)$ is  the set of minimal separators of $G$ contained
in $K$.
\end{theorem}

\section{Induced subgraph of bounded treewidth}
\label{sec:_hidden_decomposition}
In this section we prove the first result relating the problems 
of finding an induced subgraph and enumerating potential maximal cliques. 
The following  lemma is crucial for our algorithm. 

\begin{lemma}\label{le:subchordal}
Let $F=(V_F,E_F)$ be an induced subgraph of a graph $G=(V_G,E_G)$.
Then for every minimal triangulation $TF$ of $F$, there is
a minimal triangulation $TG$ of $G$ such that for every
clique $K$ of $TG$,  the intersection $K\cap V_F$ is either empty,
 or is a  clique of $TF$.
 \end{lemma}

Now we are ready to proceed with the main result of this section. 
\begin{theorem}\label{thm:mainthm_comb}
Let $G$ be a graph on $n$ vertices and $m$ edges given together with the set $\Pi_G$
of its potential maximal cliques and the set $\Delta_G$ of its minimal separators. For any integers $0\leq t,\ell\leq n$,
there is an algorithm that checks in time $\cO(n^{t+4} m(|\Pi_G|+|\Delta_G|))$
if $G$ contains an $\ell$-vertex induced subgraph  
of treewidth at most $t$.
\end{theorem}
\begin{proof}
Let $F$ be an induced subgraph of treewidth at most $t$. By
Lemma~\ref{thm:folk}, there is a minimal triangulation $TF$ of
$F$, such that the size of a maximal clique of $TF$ is at most
$t+1$. By Lemma~\ref{le:subchordal}, there is a minimal
triangulation $TG$ of $G$, such that every clique of $TG$ contains
at most $t+1$ vertices of $F$. If we knew such a minimal
triangulation $TG$, dynamic programming over the
clique-tree of $TG$ will provide the answer to our question in time $\cO(n^{t+3}m)$.
However, we are not
given such a triangulation a priori. Thus, the computations require
multiplicative factor $n|\Pi_G|$.

We start by enumerating all full blocks and sorting
them by their sizes. This can be done by enumerating all minimal
separators, and checking for each minimal separator $S$ and each
of the connected component of $G\setminus S$ if this is a full
component or not. By making use of
Theorem~\ref{th:listing_minsep}, this step can be performed in
time $\cO(|\Delta_G|  \cdot  n^{3} )$. Sorting  blocks can be done
in $\cO(n|\Delta_G|)$ time using a bucket sort.

For a minimal separator $S$, a full block  $(S,C)$,  and a potential maximal clique 
$\Omega$, we call the triple $(S,C, \Omega)$ \emph{good} if 
$S\subseteq \Omega \subseteq C\cup S$.
For each full block we also enumerate all good triples that can be obtained
from this block as follows. 
By Theorem \ref{th:pmc_sep}, if a minimal separators $S$ is a subset of  a 
potential maximal clique $\Omega$, then  $S = N(C)$ for some connected component 
$C $ of $G[V \setminus \Omega]$, and thus, the number of minimal separators contained in $\Omega$ is at most $n$. 
By Theorem \ref{th:pmc_sep}, $G\setminus \Omega$ has no full component associated to $\Omega$, and thus for every
minimal separator $S\subseteq \Omega$, we have that $\Omega\setminus S\neq \emptyset$. Therefore, 
there exists a vertex $u \in \Omega \setminus S$ 
and  thus $\Omega$  is a subset of the full block $(S,C)$ such that $u \in C$. But this yields that every potential maximal clique is 
contained in at most $n$ good triples, and the total number of good triples is  at most  $n|\Pi_G|$. 
Computing for every potential maximal clique all good triples containing it, in time 
  $\cO(m|\Pi_G|)$ one can create a data structure that 
for each full block assigns the set of potential maximal cliques that make a good triple with that block. 

\medskip 
After preprocessing  blocks and creating good triples, we proceed with dynamic programming.  The dynamic programming consists of two step. In the first, most technical step, we compute the sizes of maximal subgraphs in full blocks  $(S,C)$ subject to the condition that the minimal separator $S$ contains at most $t+1$ vertices of the subgraph. To compute these values we use deep combinatorial results of Bouchitt\'e and    Todinca on the structure of potential maximal cliques. In the second step, we go through all minimal separators, and for each separator we glue solutions found at the first step. 

\medskip
\noindent\textbf{Step~1: Processing full blocks. }
We need to define several  functions.
For a full block $(S,C)$, 
and for every subset $W \subseteq S$,
$|W|\leq t+1$, and integer $ 0 \leq \ell \leq n$,  
 $\alpha(\ell,W,S,C)=1$  if there exits an 
induced subgraph $F=(V_F, E_F)$ of 
 $G[C \cup W]$ such that $|V_F| = \ell$, $V_F\cap S =W$,  and $F$
has a minimal triangulation $TF$ such that $\omega(TF)\leq t+1$  and $W$ is 
a clique of $TF$.  Otherwise,  $\alpha(\ell,W,S,C)=0$.

For every
inclusion minimal block $(S,C)$, we have that $S \cup C$ is a 
potential maximal clique. 
Thus for every inclusion minimal block $(S,C)$, and for every
set $W \subseteq S \cup C$,
$|W|\leq t+1$,
 we put
\begin{equation*}\label{eq:base}
\alpha(\ell,W,S,C)=
 \left\{ \begin{array}{lr}
1, &  \text{  if  }\ell=|W|,\\
0,   & \text{~ otherwise.}
\end{array}\right.
\end{equation*}

To compute the values of $\alpha$ for larger blocks, we perform dynamic programming over sets of good triples formed by smaller blocks.  
For every good triple $(S,C, \Omega)$,
 and for every subset $W \subseteq \Omega$,
$|W|\leq t+1$, and integer $ 0 \leq \ell \leq n$, we want to compute an auxiliary function such that 
$\beta(\ell,W,S,C,\Omega)=1$ if there exits an 
induced subgraphs $F=(V_F, E_F)$ of $G[C \cup W]$ such that $|V_F| = \ell$, $V_F\cap \Omega =W$, and $F$
has a minimal triangulation $TF$ such that $\omega (TF) \leq t+1$, and  $W$
 is a  clique of $TF$.  Otherwise, $\beta(\ell,W,S,C,\Omega)=0$.

Let us remark that
\begin{equation*} \label{eq:alpha_beta}
\alpha(\ell,W,S,C)=1 \Leftrightarrow
\exists 
\text{ good triple } (S,C,\Omega) \text{ and } W\subseteq W' \subseteq \Omega \text{ s.t. }  \beta(\ell,W',S,C,\Omega) = 1.
 \end{equation*}  Indeed, if
$\beta(\ell,W',S,C,\Omega) = 1$, then 
 there is a minimal triangulation $TF$ of an induced subgraph $F=(V_F, E_F)$ of 
 $G[C \cup W]$ such that $|V_F| = \ell$,  $\omega(TF)\leq t+1$,  and $W$ is 
a clique of $TF$, simply because this is true for $W'$ and $W \subseteq W'$.  
Then $TF[V_F\setminus (W'\setminus W)]$ is the triangulation of 
$F[V_F\setminus (W'\setminus W)]$ that certifies $\alpha(\ell,W,S,C)=1$.
For the opposite direction 
the arguments are similar.

We start computing $\beta$ from inclusion minimal blocks. 
For every inclusion minimal block $(S,C)$, and for every
set $W \subseteq S \cup C$,
$|W|\leq t+1$,
 \begin{equation*}\label{eq:basex}
\beta(\ell,W,S,C,\Omega)=
 \left\{ \begin{array}{lr}
1, &  \text{ ~if ~}\ell=|W|,\\
0,   & \text{~ otherwise.}
\end{array}\right.
\end{equation*}

To compute $\beta(\ell, W, S, C, \Omega)$ we define an auxiliary  function $\gamma$ as follows.
Let $\{ C_1, \ldots, C_p\}$ be the vertex sets of the
connected components of $G[(S\cup C) \sm \Omega]$.
 By Theorem~\ref{th:pmc_sep},
the sets $S_i=N(C_i)$, $1\leq i\leq p$, are minimal separators
 of $G$,
and moreover, $S_i \subset \Omega$ for $1\leq i\leq p$.
 The values of  function $\gamma(\ell,j,W,S,C,\Omega)$ are in  $\{0,1\}$.  For 
 every good triple $(S,C, \Omega)$,
 and for every subset $W \subset \Omega$, 
$|W|\leq t+1$, and $ 0 \leq \ell \leq n$, 
 $\gamma(\ell,j,W,S,C,\Omega)=1$ if and only if  there exits an 
induced subgraph $F=(V_F, E_F)$ of $G[W \cup \bigcup_{i=1}^j C_i]$ 
such that $|V_F| = \ell$, $V_F\cap \Omega =W$,  and $F$ has a minimal triangulation $TF$ such that $\omega(TF)\leq t+1$ and
 $W$ is a clique in $TF$.
Note that  $G[W \cup \bigcup_{i=1}^p C_i]=G[W \cup C]$, and  by definitions of $\beta$ and $\gamma$, we have that 
\begin{equation*}\label{eq:beta}
\beta(\ell,W,S,C,\Omega) = \gamma(\ell,p,W,S,C,\Omega).
\end{equation*}

Now for every $\ell \geq 0$,

\begin{equation*}\label{eq:base2}
\gamma(\ell,1,W, S, C, \Omega)= \alpha(\ell-|W \setminus  S_1|, W\cap S_1, S_1, C_1).
\end{equation*}
For $j > 1$,
\begin{equation*}\label{eq:gamma}
\gamma(\ell,j,W, S, C, \Omega) = 
 \left\{ \begin{array}{ll}
1, & \text {if } \gamma(i,j-1,W, S, C, \Omega) = 1 \land \alpha(\ell-i +|W\cap S_j|, W\cap S_j, S_j, C_j)= 1,\\
   &  \text{for some } i, 1 \leq i \leq \ell,\\ 
0, &  \text{ otherwise}.
\end{array}\right.
\end{equation*}
 This is because  for every $\ell$-vertex subgraph $F=(V_F,E_F)$ of $G[C_1 \cup \cdots C_j \cup W]$
with  $V_F\cap \Omega =W$, 
there is $i\leq \ell$ such that $i$ vertices of $F$ are in   $C_1 \cup \cdots C_{j-1} \cup W$ and  $\ell -i + |W\cap S_j|$ vertices are in 
$C_j\cap S_j$.

To compute $\gamma(\ell,j,W, S, C, \Omega)$,  we find the blocks  
$(S_j,C_j)$,   $1\leq j \leq p$,  in $G$, which can be done in time $\cO(m)$ and read already computed  
values $\alpha(\ell-i+|W\cap S_j|, W\cap S_j, S_j, C_j)$ and $\gamma(i,j-1,W, S, C, \Omega)$.
Similarly,  the values of $\alpha(\ell,W,S,C)$ and  $\beta(\ell,W,S,C,\Omega)$ are computable in time $\cO(m)$ from the 
values of the smaller blocks and the values of $\gamma$. 
The total running time required to compute the values of 
all $\alpha(\ell,W,S,C)$  is  $\cO(m)$  times the number of 
different 6-tuple $(\ell,i,W,S,C,\Omega) $ plus the time $\cO(n^3 (|\Delta_G| +|\Pi_G|))$ required for preprocessing step. The  number of good triples $(S,C,\Omega)$ is at most 
$n|\Pi_G|$, and the number of subsets $W$ of size at most $t+1$ is $\cO(n^{t+1})$.
Thus the total running time required to compute all values  $\alpha(\ell,W,S,C)$ is
\[
\cO(mn^{t+4} (|\Pi_G|+|\Delta_G|)).
\]

\medskip

Now everything is prepared to  solve the problem on  graph $G$ and to conclude the proof.
 By Lemma~\ref{le:subchordal},  if $F$ is an induced subgraph of $G$ of treewidth at most $t$, 
 there exists a minimal separator $S$ of $G$, such 
that $|V_F \cap S| \leq t+1$.  We go through all minimal separators, and for each minimal separator $S$, we try to glue solutions obtained during the first step.

\medskip
\noindent\textbf{Step~2: Gluing pieces together.}
Let $S$ be a minimal separator and let $\{ C_1, \ldots, C_p\}$ be the vertex sets  of the
 connected components of $G[V \sm S]$.  We put $S_i = N(C_i)$. % and $W_i= W \cap S_i$.
For every subset $W\subseteq S$ of size at most 
$t+1$, and integer $ 0 \leq \ell \leq n$, we 
define $\delta(\ell,j,W,S)=1$  if there is an 
induced $\ell$-vertex subgraph $F=(V_F, E_F)$ of 
$G[W \cup \bigcup_{i=1}^j C_i]$  which poses 
 a minimal triangulation $TF$ with $\omega(TF) \leq t+1$, and
  such that  $W=V_F\cap S$ is a clique in $TF$.   If no such graph $F$ exists, we put  $\delta(\ell,j,W,S)=0$. 
By Lemma~\ref{le:subchordal},   $G$ has an induced $\ell$-vertex subgraph of treewidth at most $t$ if and only if 
$\delta(\ell,p,W,S) = 1$  for some minimal separator $S$. Thus computing the value $\delta$ for all minimal separators is sufficient for deciding if $G$ has an induced subgraph on $\ell$ vertices of treewidth at most $t$.

For every $\ell\geq 0$ and $j = 1$, we have that
\begin{equation*}\label{eq:base3}
\delta(\ell,1,W, S)= \alpha(\ell-|W \setminus S_1|, W\cap S_1, S_1, C_1).
\end{equation*}

For $j > 1$,

\begin{equation*}\label{eq:mu}
\delta(\ell,j,W, S) = 
 \left\{ \begin{array}{ll}
1, &\text{if }  \delta(i,j-1,W, S) = 1 \wedge
   \alpha(\ell-i+|W\cap S_j|, W\cap S_j, S_j,C_j)= 1,\\
   & \text{for some }  1 \leq i \leq \ell,\\
0, &  \text{otherwise}.
\end{array}\right.
\end{equation*}
Like in the case with $\gamma$, the correctness of the formula above follows from the fact, that 
for every $\ell$-vertex subgraph $F=(V_F,E_F)$ of $G[C_1 \cup \cdots C_j \cup W]$
with  $V_F\cap S=W$, 
there is $i\leq \ell$ such that $i$ vertices of $F$ are in   $C_1 \cup \cdots C_{j-1} \cup W$ and  $\ell -i + |W\cap S_j|$ vertices are in 
$C_j\cap S_j$.

Concerning the time required to perform this step. 
Like in above, in time $\cO(m)$ we can  find the connected components $\{ C_1, \ldots, C_p\}$
 of $G[V \sm S]$, and the corresponding  full blocks $(S_i,C_i)$. Thus the running  of this step is proportional to 
 $m$ times the number of $4$-tuples  $(\ell,j,W, S)$, and we conclude that this step of the algorithm can be performed in time 
  $\cO( mn^{t+3} \cdot |\Delta_G| )$.
\end{proof}

\section{Induced subgraph isomorphism} \label{sec:modification}
The technique described in the previous section with slight modifications can be applied for many different problems. In this section we give an important example  of such modification. 
%Our first example is the problem of finding an induced subgraph in a graph $G$ which is isomorphic to a given graph $F$. 

\begin{theorem}\label{thm:isubiso}
Let $G$ be an $n$-vertex graph  given together with the set $\Pi_G$
of its potential maximal cliques and the set $\Delta_G$ of its minimal separators. 
Let $F$ be a graph of treewidth $t$.  There is an algorithm 
checking if $G$ contains an induced subgraph isomorphic to  $F$ in time  
$\cO(n^{\cO(t)}(|\Delta_G|+|\Pi_G|))$. 
\end{theorem}
\begin{proof}
The proof of the theorem follows the lines of Theorem~\ref{thm:mainthm_comb} with modifications that are similar to the well known Bodlaender's algorithm for solving the graph isomorphism problem on graphs of bounded treewidth \cite{Bodlaender:1990hb}. We outline only the most important differences of such a  modification.

The treewidth of $F $ is at most $t$, and we use the algorithm of Arnborg et.al. \cite{Arnborg:1987si} to construct a
minimal triangulation $TF$ of $F$ such that  $\omega(TF) \leq t+1$. The running time of this algorithm is  
in $\cO(n^{t+2})$. The number of maximal cliques and minimal separators  in an $n$-vertex chordal graph is $\cO(n)$
\cite{RoseTL76}. Thus the number of full blocks and good triples in  $TF$ is $\cO(n)$. We list and keep all these blocks and triples.  This   can be done in polynomial time. 

As in the proof of Theorem~\ref{thm:mainthm_comb}, we perform two steps of dynamic programming. First we run computations over  full blocks of $G$, and then use computed values to glue solutions in minimal separators. 

For every full block $(S,C)$ of $G$, every full block $(S_F,C_F)$ of $TF$, every 
 subset $W \subseteq S$, where $|W|= |S_F|\leq t+1$, and every bijection $\mu \colon S_F \to W$, we define the value 
 $\alpha(S_F,C_F,W,\mu, S,C)$
to be equal to $1$ if there is an injection $\lambda \colon S_F\cup C_F \to W\cup C$ such that $F[S_F\cup C_F]$ is isomorphic to $G[\lambda(S_F\cup C_F)]$, and for every $v\in S_F$, $\lambda (v)= \mu(v)$. Otherwise, we put $\alpha(S_F,C_F,W,\mu, S,C)=0$. In other words, $\alpha$ is equal to 1, when  $G[W\cup C]$ contains a subgraph isomorphic to $F[S_F\cup C_F]$, and moreover, the restriction of the corresponding isomorphic mapping  on $S_F$ is exactly $\mu$. 
 
As in Theorem~\ref{thm:mainthm_comb}, to compute $\alpha(S_F,C_F,W,\mu, S,C)$ we run through good triples 
$(S,C,\Omega)$, where $\Omega$ is a potential maximal clique, $S\subseteq \Omega \subseteq S\cup C$.
For every good triple $(S,C, \Omega)$ of $G$ and every good triple $(S_F,C_F, \Omega_F)$ of $F$, 
 for every subset $W \subseteq \Omega$, such that  $|W|=|\Omega_F|\leq t+1$, and  every bijection
 $\mu \colon \Omega_F \to W$, we define the function 
 $ \beta(S_F,C_F,\Omega_F, W,\mu, S,C, \Omega)\in \{0,1\}$. 
 We put  $ \beta(S_F,C_F,\Omega_F, W,\mu, S,C, \Omega)=1$ if and only if  
   there is an injection $\lambda \colon S_F\cup C_F \to W\cup C$ such that $F[S_F\cup C_F]$ is isomorphic to $G[\lambda(S_F\cup C_F)]$, and for every $v\in \Omega_F$, $\lambda (v)= \mu(v)$.  Following the lines of Theorem~\ref{thm:mainthm_comb}, it is possible to show that $\alpha(S_F,C_F,W,\mu, S,C)=1$ if and only if there exist
   \begin{itemize}
   \item Good triple $(S,C, \Omega)$ of $G$ and  good triple $(S_F,C_F, \Omega_F)$ of $F$;
   \item Set $W'$, $W\subseteq W'\subseteq \Omega$;
   \item Bijection $\mu ' \colon \Omega_F \to W'$, $\mu'_{|W}(\cdot ) =\mu(\cdot ) $
 \end{itemize}
 such that  $ \beta(S_F,C_F,\Omega_F, W',\mu', S,C, \Omega)=1$. 
 
 The main difference with the proof of Theorem~\ref{thm:mainthm_comb} is in the way we compute $\beta$. We compute the values of  $ \beta(S_F,C_F,\Omega_F, W,\mu, S,C, \Omega)$ from the values  of smaller blocks contained in $G[S\setminus \Omega]$. This is done by reducing to the problem of finding a maximum matching in some auxiliary  bipartite graph. This step is quite similar to the algorithm of Bodlaender \cite{Bodlaender:1990hb} for isomorphism of bounded treewidth graphs. 
 Let $F_1, F_2, \dots, F_p$ be the connected components of 
 the graph $F[C_F\setminus \Omega_F]$. 
 Then the sets $Q_i =N_F(F_i) \subseteq \Omega_F$ are minimal separators and pairs $(F_i, Q_i)$,  $1\leq i\leq p$, 
are blocks in $F$.  Similarly, 
for  the connected components  $G_1, G_2, \dots, G_q$  of 
   $G[C\setminus \Omega]$, we put $S_i=N_G(G_i)$, and define blocks  $(G_i, S_i)$, $1\leq i\leq q$. 
   We construct an auxiliary bipartite graph $B$ with bipartition  $X=\{x_1, x_2, \dots, x_p\}$ and   $Y=\{y_1, y_2, \dots, y_q\}$. 
   There is an edge $\{x_i, y_j\}$ in $B$ if and only if there is an isomorphic mapping of block  $(F_i, Q_i)$ to  block $(G_j, Q_j)$
which agrees with $\mu$. But then to decide if blocks $(F_i, Q_i)$ can be mapped to blocked $(G_i, S_i)$ is equivalent to deciding if $B$ has a matching of size $p$. 
More formally,   $\{x_i, y_j\}$ is an edge in $B$ 
if and only if there is an injection $\lambda \colon F_i\cup Q_i \to G_j\cup S_j$ such that $F[F_i\cup Q_i ]$ is isomorphic to $G[\lambda(F_i\cup Q_i )]$, and for every $v\in Q_i$, $\lambda (v)= \mu(v)$. But such an injection $\lambda$ exists if and only if 
$\alpha(F_i,Q_i,W',\mu', G_j,S_j)=1$, where $W'=\mu(Q_i)$ and $\mu' (\cdot) = \mu_{|Q_i}(\cdot)$. Therefore, to compute the value of $\beta$, it is sufficient to run through the already computed values of $\alpha$ of smaller blocks, construct an auxiliary graph and find if this graph contains a matching of specific size. 

Finally, as in Theorem~\ref{thm:mainthm_comb}, after all values $\alpha$ are computed, we run through all minimal separators of $G$ 
 and for each minimal separator $S$, we try to glue solutions obtained   for all blocks attached to this separator. Here again, we need only the values of $\alpha$ computed for all such blocks and   reduce the problem to bipartite matchings. 
The running time of the algorithm is up to  multiplicative polynomial factor equal to the number of states of the dynamic programming. 
To compute the values of $\alpha$ and $\beta$, we run through all  potential maximal cliques,  blocks, and  good triples of $TF$ and $G$, which is $n^{\cO(1)}|\Pi_G|$.  For every pair of blocks or triples, we run through all subsets $W$ of size at most $t+1$, which is $\cO(n^{t+1})$, and through all mappings between sets of cardinality at most $t+1$, which is $\cO((t+1)^{t+1})$. Finally, we run through all minimal separators. Thus the total running time of the algorithm is $\cO(n^{\cO(t)}(|\Delta_G|+|\Pi_G|))$.     The proof of the correctness of the algorithm follows the lines of  Theorem~\ref{thm:mainthm_comb},    
      and we omit it here. 
%: stop here: 
\end{proof}
Let us also remark that with a standard bookkeeping, the algorithm of  Theorem~\ref{thm:isubiso} can also output a subgraph of $G$ isomorphic to $F$.

\section{Enumerating  potential maximal cliques}\label{sec:faster_pmc}

In this section we show that all \pmcs\ of  graph $G=(V,E)$ can be enumerated by making use of 
connected vertex sets with special restrictions. This approach represents a significant simplification
over previous algorithms for listing \pmcs\ \cite{FominKratschTodincaV08,FominV08}.
More precisely,  we show that for every \pmc\ $\Omega$ there exists a vertex set
$Z \subset V$ and a vertex $z \in Z$ such that
\begin{itemize}
 \item $|Z|-1 \leq (2/3)(n-|\Omega|)$,
 %\item $z \in Z$,
 \item $G[Z]$ is connected,
 \item $\Omega = N(Z \setminus \{z\})$ or $\Omega = N(Z) \cup \{z\}$.
\end{itemize}

 As far as we obtain such a classification, the enumeration algorithm is extremely simple: 
For each vertex $z \in V$ enumerate
every connected vertex set $Z$ containing $z$ 
where $|Z|-1 \leq 2|V \sm N[Z-\{z\}]|$. (In other words we test for each 
connected vertex set $Z$ containing $z$, where at least $\frac{|Z|-1}{2}$ vertices are 
not contained in $N[Z \sm \{z\}]$.)
For each of these subsets, we run the algorithm  of
Bouchitt\'e and    Todinca  from \cite{BouchitteT01} to check 
if
$N(Z \setminus \{z\})$ or $N(Z) \cup \{z\}$ is a \pmc.  The algorithm of  Bouchitt\'e and    Todinca checks in  $\cO(nm)$ time if a 
vertex set $\Omega$  is a  \pmc. 
This  is a significant simplification comparing to previous enumeration algorithms  \cite{FominKratschTodincaV08,FominV08}
 avoiding complications  with  different treatments  of   nice and (not) nice
\pmcs.

\medskip 

We proceed with a sequence of technical lemmas.
For a \pmc\ $\Omega$ and a vertex $x \in \Omega$ we define by $D_x$  the vertex sets of all connected components
$C$ of $G[V \setminus \Omega]$  with  $x \in N(C)$.

\begin{lemma}\label{le:edge}
Let $\Omega$ be a \pmc\ of $G=(V,E)$, and let $\{x,y\}$ be an edge of $G[\Omega]$ such that $\Omega$ is not a \pmc\ in $G \setminus \{x,y\}$. Then there is  $Z\subseteq V$ and  $z\in Z$, such that
\begin{itemize}
 \item $\Omega = N(Z) \cup \{z\}$,
 \item $G[Z]$ is connected, and
% \item $z \in Z$, and
 \item $|Z|-1 \leq (1/2)(n-|\Omega|)$.
\end{itemize}
\end{lemma}
% \begin{proof}
% The proof has been moved to Appendix. 
% % Vertex sets $D_x$ and $D_y$ are  disjoint because otherwise the existence of a connected component of $G[V \setminus \Omega]$ containing both $x$ and $y$ yields that $\Omega$ is  a \pmc\ in $G \setminus \{x,y\}$. Let us assume that $|D_x| \leq |D_y|$. 
% % We put $Z:=D_x \cup \{x\}$ and $z:=x$.
% \end{proof}

\begin{corollary}\label{cor:full}
Let $\Omega$ be a \pmc\ of $G=(V,E)$, such that 
$\Omega$ is a \pmc\ in $G\setminus \{x,y\}$ for every edge $\{x,y\}$ of $ G[\Omega]$. 
Then $N(D_x) = \Omega$ for every vertex $x \in \Omega$.
\end{corollary}
% \begin{proof}
% The proof has been moved to Appendix. 
% % Theorem~\ref{th:pmc_sep} implies that either $\{x,y\}$ is an edge of $G$, 
% % or there exists a connected component $C$ of $G[V \sm \Omega]$ such that $x,y \in N(C)$.
% % Since $\Omega$ is a \pmc\ in $G\setminus \{x,y\}$ for every edge $\{x,y\}$ of $ G[\Omega]$
% % we conclude that $x,y \in N(C)$ for some connected component $C$ of $G[V \sm \Omega]$, and 
% % thus $N(D_x) = \Omega$ for every vertex $x \in \Omega$.
% \end{proof}

Let $\mathcal{C}$ be the set of connected components of $G[V \setminus \Omega]$ with the following two properties:
For each connected component $C\in \mathcal{C}$ there exists a pair of vertices $x,y \in \Omega$
such that $C$ is the unique component from $\mathcal{C}$  with  $x,y \in N(C)$, and 
for each  pair of vertices $x,y \in \Omega$ there exists a connected component $C\in \mathcal{C}$ such that 
$x,y \in N(C)$. Let $W$ be the vertex set of $\mathcal{C}$, we refer to the graph  $G' = G[\Omega \cup W]$ 
as to a \emph{reduced graph for $\Omega$}. 
In other words $\mathcal{C}$ is an inclusion minimal witness 
for $\Omega$ being a \pmc\ of $G$, by only using connected components of $G[V \sm \Omega]$.
The set $\mathcal{C}$  can be constructed  by the following procedure which is repeated recursively if possible: 
If there exists a connected component $C$ of $G[V \setminus \Omega]$ 
such that for each pair $x,y \in N(C)$ there is a connected component $C' \neq C$ in $G[V \setminus \Omega]$ 
such that $x,y \in N(C')$, then remove $C$ from the graph.

\begin{lemma}\label{le:4comp}
Let $\Omega$ be a \pmc\ of $G=(V,E)$ such that $\Omega$ is also a 
\pmc\ in $G\setminus \{x,y\}$ for every edge $\{x,y\}$ of $ G[\Omega]$, and 
where $G' = G[\Omega \cup W]$ contains at least $4$ connected components. 
Then there is $Z\subset V$ and $z\in Z$ such that
\begin{itemize}
 \item $\Omega = N(Z \setminus \{z\})$,
 \item $G[Z]$ is connected, and
 %\item $z \in Z$, and
 \item $|Z|-1 \leq (3/5)(n-|\Omega|)$.
\end{itemize}
\end{lemma}

The following  characterization is used in the  new algorithm enumerating  \pmcs.

\begin{lemma}\label{le:small_set}
For every \pmc\ $\Omega$ of $G=(V,E)$, there exists a vertex set $Z\subseteq V$ and   $z\in Z$ such that
\begin{itemize}
 \item $|Z|-1 \leq (2/3)(n-|\Omega|)$,
% \item $z \in Z$,
 \item $G[Z]$ is connected, and
 \item $\Omega = N(Z \setminus \{z\})$ or $\Omega = N(Z) \cup \{z\}$.
\end{itemize}
\end{lemma}
% \begin{proof}
% The proof has been moved to Appendix. 
% % If there are adjacent vertices $x,y \in \Omega$ such that  $\Omega$ is not a \pmc\ in
% % $G\setminus \{x,y\}$, then by Lemma~\ref{le:edge},  there exists a connected vertex set $Z$ and   $z \in Z$ such that $\Omega = N(Z) \cup \{z\}$ and $|Z| -1 \leq (1/2)(n- |\Omega|)$.
% % 
% % 
% % Thus, we can assume that for every pair of adjacent vertices $x,y \in \Omega$,  $\Omega$ is also a  \pmc\ in
% % $G\setminus \{x,y\}$.
% % In the case where $G'[V' \setminus \Omega]$ contains more than three connected components, 
% % it follows by Lemma~\ref{le:4comp} that the claim of the lemma holds, since $3/5 < 2/3$. 
% % Only remaining case when $G'[V' \setminus \Omega]$ contains exactly three connected components. 
% % Let $C_1,C_2,C_3$ be the connected components, and let $|C_1| \leq |C_2| \leq |C_3|$.
% % By Theorem~\ref{th:pmc_sep}, we know that $N(C_i) \subset \Omega$ for $i \in \{1,2,3\}$.
% % Thus, there exists a vertex $z \in \Omega \setminus N(C_3)$. 
% % Then $D_z = C_1 \cup C_2$  and  by putting  $Z:=D_z \cup \{z\}$, 
% % we obtain the desired set $Z$ and vertex $z\in Z$.
% \end{proof}
Let us remark that {Lemma}~\ref{le:small_set} yields a simple algorithm enumerating potential maximal cliques. 
We just connected vertex sets $Z$ of bounded size and  check if either
$N(Z \setminus \{z\})$ or $  N(Z) \cup \{z\}$  is a \pmc. The enumeration of such connected vertex sets can be done in time 
 $\cO(n^2 \cdot  \constOldPMC^n)$ \cite{FominV08} and checking if a set is 
 a \pmc\ in $\cO(nm)$ time \cite{BouchitteT01}.
% 
% of the two cases give a \pmc\ in $\cO(nm)$ time \cite{BouchitteT01}.
%By the result of \cite{FominV08} on the number of connected vertex sets of bounded size, 
%all these candidate sets can be listed in $\cO(n^2 \cdot  \constOldPMC^n)$ time, and it can 
%be checked if either of the two cases give a \pmc\ in $\cO(nm)$ time \cite{BouchitteT01}.

In what follows we improve (slightly) the running time of the algorithm. 
The improvement is based on the previous lemmata.
The proof gain
by exploiting the fact that the most time consuming case is when 
there are exactly three connected components in the reduced graph. 

\begin{theorem}\label{thm:pmc_modified}
All \pmcs\ of an $n$-vertex graph can be enumerated in time $\cO(\constPMC^n)$.
\end{theorem}

We need the following results.

\begin{theorem}[Berry, Bordat, and Cogis \cite{BerryBC00}]\label{th:listing_minsep}
There is an algorithm listing all  minimal separators of an input
graph $G$ in $\cO(n^3|\Delta_G|)$ time.
\end{theorem}
%: QQQ

\begin{theorem}[Fomin and Villanger \cite{FominV08}]\label{th:number_minsep}
Every $n$-vertex graph has 
$\cO(1.6181^n)$ minimal separators. 
\end{theorem}

Putting together  Theorems~ \ref{thm:mainthm_comb},  \ref{thm:pmc_modified},  \ref{th:listing_minsep}, and~\ref{th:number_minsep},  we arrive at the following corollary. 

\begin{corollary}
For every $t\geq 0$,   a maximum induced subgraph of treewidth at most $t$ in an $n$-vertex graph $G$ can be found in time 
 $\cO( \constPMC^n \cdot n^{\cO(t)})$. 
\end{corollary}

Similarly,  by Theorem~\ref{thm:isubiso}, 
we have the following corollary. 
\begin{corollary}
For every $t\geq 0$ and graph $F$ of treewidth $t$, checking if an $n$-vertex graph $G$ contains an induced subgraph isomorphic to $F$ (and finding one if such exist) can be done  
 in time   $\cO( \constPMC^n \cdot n^{\cO(t)})$. 
\end{corollary}
 Let us remark that the treewidth of an $n$-vertex planar, and more generally, graph excluding some fixed graph as a minor, is 
 $\cO(\sqrt{n})$ \cite{AlonST90}. Therefore, if $F$ is a graph excluding some fixed graph as a minor, deciding if $G$ has induced subgraph isomorphic to $F$ can be done in time  $ \constPMC^{n +o(n)}$. 
% As a consequence we can find 
%maximum induced regular graphs, like grids since there is at most a polynomial number of possible maximal graphs. 

\section{Conclusion and open questions} \label{sec:conclusion}
In this paper we have shown how the theory of minimal triangulations can be used to obtain  moderate exponential algorithms for a number of problems about induced subgraphs. 
With some modifications our technique can be used for different problems of the same flavor, like finding a maximum connected induced subgraph of small  treewidth. %The dynamic programming over blocks and good triples used in our algorithm has some similarities with classical dynamic programming on graphs of bounded treewidth.   By Courcelle's Theorem \cite{Courcelle90}, we have a good understanding  of what kind of problems can be solved efficiently on graphs of bounded treewidth. 
 It would be interesting to see  if Theorem~ \ref{thm:mainthm_comb} can be extended   for finding  maximum induced subgraphs with other specific properties like being planar or excluding some $h$-vertex graph $H$ as a minor. 
 %The difficulty here is that a naive modification of our dynamic programming   

 Another very interesting question is, how many potential maximal cliques can be in an $n$-vertex graph? Theorem~\ref{thm:pmc_modified}  says that roughly at most $ \constPMC^n $.  How tight is this bound?  There are graphs
with  roughly $3^{n/3} \approx 1.442^n$ potential maximal cliques  
 \cite{FominKratschTodincaV08}. Let us remind that by the classical result of Moon and Moser \cite{MoonM65} 
(see also  Miller and Muller \cite{MillerMuller60}) 
that the number of maximal cliques  in a graph on $n$ vertices
is at most $3^{n/3}$. Can it be that the right upper bound on the number of potential maximal cliques 
 is also  roughly $3^{n/3}$?   By Theorem~ \ref{thm:mainthm_comb}, this would yield a dramatic improvement for many moderate exponential algorithms. 

\end{document}